\documentclass[a4paper,anonymous, UKenglish, cleveref, autoref, thm-restate]{lipics-v2019}

\usepackage{amsthm}
\usepackage{algorithm, algorithmic}
\usepackage{graphicx}

\newtheorem{thm}{Theorem}[section]
\theoremstyle{definition}

\theoremstyle{remark}

\theoremstyle{plain}
\newtheorem{lem}[thm]{Lemma}
\newtheorem{col}[thm]{Corollary}

\title{Sleeping Model: Local and Dynamic Algorithms}

\author{Tzalik Maimon}{Ben-Gurion University of The Negev, Israel}{tzalik@post.bgu.ac.il}{}{}

\authorrunning{T. Maimon}

\Copyright{T. Maimon}

\ccsdesc[500]{Computing methodologies~Distributed algorithms}

\keywords{Distributed Computing, Sleeping model, Bounded Neighborhood Independence}

\acknowledgements{I want to thank Leonid Barenboim for the fruitful discussions. I want to thank Nofit Fraizait for the great feedback.}

\EventNoEds{2}
\EventLongTitle{Conference on Principles of Distributed Systems}
\EventShortTitle{OPODIS'21}
\EventAcronym{OPODIS}
\EventYear{2021}
\EventDate{December 13--15, 2021}
\EventLocation{Strasbourg, France}
\EventLogo{}

\begin{document}
\nolinenumbers
\maketitle

\begin{abstract}
The distributed setting represents a communication network where each node is a processor with its own memory, and each edge is a direct communication line between two nodes. Algorithms are usually measured by the number of synchronized communication rounds required until the algorithm terminates in all the nodes of the input graph. In this work we refer to this as the {\em clock round complexity}. The setting is studied under several models, more known among them are the $\mathcal{LOCAL}$ model and the $\mathcal{CONGEST}$ model. In recent years the {\em sleeping model} (or some variation thereof) came to the focus of researchers \cite{BT19, CGP20, DHW98, F20, GKKPS08, KPSY11, BM21}. In this model nodes can go into a sleep state in which they spend no energy but at the same time cannot receive or send messages, nor can they perform internal computations. This model captures energy considerations of a problem. In \cite{BM21} Barenboim and Maimon defined the class of {{\bf O-LOCAL}\footnote{A problem $P$ is an {\bf O-LOCAL} problem if, given an acyclic orientation on the edges of the input graph, one can solve the problem as follows. Each vertex awaits the decisions of its parents according to the given orientation and can make its own decision in regard to $P$ using only the information about its parents decisions.}} problems and showed that for this class of problems there is a deterministic algorithm that runs in $O(\log \Delta)$ awake time. The clock round complexity of that algorithm is $O(\Delta^2)$. Well-studied {\bf O-LOCAL} problems include coloring and Maximal Independence Set (MIS). 

In this work we present several deterministic results in the sleeping model:
\begin{enumerate}
    \item We offer three algorithms for the {\bf O-LOCAL} class of problems with a trade off between awake complexity and clock round complexity. One of these algorithms requires only $O(\Delta^{1+\epsilon})$ clock rounds for some constant $\epsilon>0$ but still only $O(\log \Delta)$ awake time which improves on the algorithm in \cite{BM21}. We add to this two other algorithms that trade a higher awake complexity for lower clock round complexity. We note that the awake time incurred is not that significant.
    \item We offer dynamic algorithms in the sleeping model. We show three algorithms for solving dynamic problems in the {\bf O-LOCAL} class as well as an algorithm for solving any dynamic decidable problem.
    \item We show that one can solve any {\bf O-LOCAL} problem in constant awake time in graphs with constant neighborhood independence. Specifically, our algorithm requires $O(K)$ awake time where $K$ is the neighborhood independence of the input graph. Graphs with bounded neighborhood independence are well studied with several results in recent years for several core problem in the distributed setting \cite{BE13,MS20,BM18, BM19, BM20}.
\end{enumerate}

\end{abstract}

\section{Introduction}

In the distributed setting there are several problems which take a central place in the attention of researchers. Among those are the Maximal Independent Set (MIS), Maximal Matching and graph coloring. These problems have been studied in various distributed models like the $\mathcal{CONGEST}$ model, the congested clique and, recently, the sleeping model. The sleeping model is a distributed model where network nodes can enter a sleeping state. As long as a vertex is in the awake state, it can do whatever it can in the standard message-passing setting. However, when a vertex is asleep it cannot receive or send messages in the network nor can it perform internal computations. On the other hand, sleeping rounds do not count towards the complexity in the model which is aptly named {\em awake complexity.} This complexity is expressed by the number of awake rounds a vertex spends during an execution. In this paper we offer several results in the sleeping model. First, we offer three algorithms with different awake time complexity and clock round complexity for solving any {\bf O-LOCAL} problem. This class of problems is consisted of problems that, given an acyclic orientation, can be solved as follows. Each vertex awaits all its parents (neighbors on outgoing edges) to calculate their solution for $P$ and then decides on its own solution for the problem. The three algorithms are summarized in Table \ref{tab:OLOCAL}.

\begin{table}[H]
\begin{center}
  \begin{tabular}{|l|l|l|}
    \hline
        {\bf Method} & {\bf Awake Time} & {\bf Clock Rounds} \\
    \hline
        Theorem \ref{thm: BEvar} & $O(\log^2 \Delta + \log^*\Delta \log^*n)$ & $O(\Delta + \log^*\Delta \log^*n)$ \\
    \hline
        Theorem \ref{thm:KWsub} & $O(\log^2 \Delta + \log^*n)$ & $O(\Delta \log \Delta + \log^* n)$ \\
    \hline
        Theorem \ref{thm:KWdebts} & $O(\log \Delta + \log^*n)$ & $O(\Delta^{1+\epsilon} + \log^* n)$ \\
    \hline
  \end{tabular}
  \caption{:  {\bf O-LOCAL} Algorithms: the trade-off between awake complexity and clock round complexity.}  \label{tab:OLOCAL}
\end{center}
\end{table}

Under this model one can also consider dynamic algorithms in which the task is to update a given solution for a problem after an {\em update phase} in which changes are made to the input graph. These changes can be expressed as updates to the vertex set or the edge set of the graph. In this setting it is common practice to analyze the performance of the update phases and the preparation phase separately where usually the update time takes priority. This is since the assumption is that the number of changes is much greater over time than the complexity of solving the problem in an non-dynamic setting. In this paper we will show that in the sleeping model, with the solution we show here, one achieves dynamic algorithms for decidable problems as well as the {\bf O-LOCAL} class of problems. Our results appear in Table \ref{tab:dynLocal}.

\begin{table}[H]
\begin{center}
  \begin{tabular}{|l|l|l|}
    \hline
        {\bf Method} & {\bf Awake Time} & {\bf Clock Rounds} \\
    \hline
        Theorem \ref{thm: BEvar} & $O(\log^2 \alpha + \log^*\alpha \log^*\beta)$ & $O(\alpha + \log^*\alpha \log^*\beta)$ \\
    \hline
        Theorem \ref{thm:KWsub} & $O(\log^2 \alpha + \log^*\beta)$ & $O(\alpha \log \alpha + \log^* \beta)$ \\
    \hline
        Theorem \ref{thm:KWdebts} & $O(\log \alpha + \log^*\beta)$ & $O(\alpha^{1+\epsilon} + \log^* \beta)$ \\
    \hline
  \end{tabular}
  \caption{: {\bf O-LOCAL} Dynamic algorithms and the trade off between awake complexity and clock round complexity.}  \label{tab:dynLocal}
\end{center}
\end{table}

Another result we offer in this work is a a constant awake time algorithm for solving any {\bf O-LOCAL} problem in graphs with bounded neighborhood independence. The family of graphs with bounded neighborhood independence is a very wide family of dense graphs. In particular, graphs with constant neighborhood independence include line-graphs, claw-free graphs, unit disk graphs, and many other graphs. Thus, these graphs represent very well various types of networks. Our technique for this family of graphs is to build an acyclic orientation of the input graph in constant awake time such that each vertex has constant number of parents in the orientation. This allows us to solve {\bf O-LOCAL} problems in constant time according to the definition of this class.

\subsection{Previous Work}

Several variants of the sleeping setting were studied recently \cite{BT19, CGP20, DHW98, F20, GKKPS08, KPSY11}. The particular setting and complexity measurement we consider in this paper were introduced by Chatterjee, Gmyr and Pandurangan \cite{CGP20} in PODC'20. In this paper the authors presented a Maximal Independent Set randomized algorithm with expected awake complexity of $O(1)$. Its high-probability awake complexity is $O(\log n)$, and its worst-case awake complexity is polylogarithmic. Recently Barenboim and Maimon \cite{BM21} showed a completeness on the class of decidable problems with a tight bound of $\Theta(\log n)$ awake time complexity. That is, they offered an algorithm for solving any decidable problem in the sleeping model in $O(\log n)$ awake time but also showed a specific decidable problem which requires at least $\Omega(\log n)$ awake time.  Furthermore, for the {\bf O-LOCAL} class of problems, among them the problem of MIS, they showed an even better algorithm with $O(\log \Delta)$ awake time. The importance of the sleeping model was shown empirically in experiments \cite{FN01, ZK05}. The energy consumption in an idle state is only slightly smaller than when a node is active. However, in the sleeping state the energy consumption is decreased drastically. Thus, if a node may enter a sleeping mode to save energy during the course of an algorithm, one can significantly improve the energy consumption of the network during the execution of an algorithm.

\subsection{Our Technique}

In \cite{BM21} the authors showed that given an acyclic orientation on the edge set of the input graph $G$, one can use the given orientation to build a binary tree internally in each vertex $v$ and have $v$ in the awake state exactly when one of its parents in the orientation sends information to $v$. This is facilitated to solve any {\bf O-LOCAL} problem. In their paper, Barenboim and Maimon used a $O(\Delta^2)$-vertex-coloring to achieve this required acyclic orientation but we here capitalize on this idea and offer different ways to build an initial orientation such that, in terms of awake complexity, our algorithms are more efficient. 

For graphs with bounded neighborhood independence we do the same by showing that building an initial orientation takes $O(1)$ time in this family of graphs. Furthermore, we show that the length of the orientation, which directly effects the efficiency of the algorithms in both awake time and number of clock rounds required, is at most $K+1$ where $K$ is the neighborhood independence of $G$. This gives us an efficient scheme for solving any {\bf O-LOCAL} problem in constant awake time which, up to a constant, is most efficient.

\section{Preliminaries}

\noindent {\bf The O-LOCAL Class.} This is a class of problems that can be solved in the following way. Given an acyclic orientation of the edges and a problem $P$, each vertex awaits the decisions of its parents in the orientation and is able then to make a decision and find a solution for itself in regard to problem $P$ using only internal computations. Core problems in the distributed setting that belong to this class are the MIS and coloring problems. \\

\noindent {\bf Awake Worst-Case Complexity.} This measurement is defined as the worst-case number of rounds in which any single vertex is awake. That is, if $w(v)$ is the number of awake rounds that a vertex $v \in G$ is awake, then the awake worst case complexity is defined as $\max(w(v)|v\in G)$. In this work, we loosely use the term {\em awake complexity} to mean awake worst-case complexity. \\

\noindent {\bf Neighborhood Independence.} Given a vertex $v \in G$, the neighborhood independence of $v$ is the maximum size set of independent neighbors one can build from the 1-hop neighborhood of $v$. The neighborhood independence of $G$ is the maximum among all neighborhood independence values of the vertices in $G$.

\section{Trade-Offs for {\bf O-LOCAL} Problems}   \label{sec:tradeOffs}

Distributed algorithms in the $\mathcal{LOCAL}$ model have a certain running time which depends on vertices which await their turn to make a decision. On the contrary, in the sleeping model these vertices do not count towards the worst-case awake time of the algorithm. But one might still wish to achieve efficiency in both measurements, that is, both the energy required for the algorithm to terminate as well as the number of total number of communication clock rounds. And thus, algorithms that offer trade offs between these two measurements are of interest. Here we bring three such algorithms and show this trade off can be achieved. Furthermore, our result from Theorem \ref{thm:KWdebts} improves on the algorithm shown in \cite{BM21} for the {\bf O-LOCAL} class of problems regardless of the trade off mentioned. Specifically, the number of clock rounds is the same while we improve on the worst case awake time. The results of this section are summarized in Table \ref{tab:OLOCAL}.

For the purpose of this section we will utilize the algorithm suggested in \cite{BM21} for solving any {\bf O-LOCAL} problem in $O(\log \Delta + \log^*n)$ awake time.  In this algorithm, we start by coloring the vertices and using this coloring we build binary trees in each vertex according which the vertex awakes and sleeps. The depth of these trees is $O(\log \Delta)$ which sets the awake time complexity of the algorithm. The number of colors used (originally $O(\Delta^2)$ in \cite{BM21}) sets the number of clock rounds required for using these binary trees to solve an {\bf O-LOCAL} problem. Given an initial coloring of the graph with $d$ colors, it is obvious one can orient the edges of the graphs to achieve an acyclic orientation of length $d$. Using the algorithm described above, from now denoted $A$, one can solve any {\bf O-LOCAL} problem in $O(\log d)$ awake time and $O(d)$ clock rounds. We will use this property of the algorithm $A$ for iterative re-coloring of the input graph. This method reduces the number of clock rounds required while maintaining low energy consumption in the network. We state the {\bf performance property of the algorithm $A$} as it plays an important role in the analysis of the algorithms described in this section.\\

\noindent {\bf The Performance Property of Algorithm $A$.} Given a graph $G$ with $d$-vertex-coloring, the algorithm $A$ can be used for solving any {\bf O-LOCAL} problem in $O(\log d)$ awake time and $O(d)$ clock rounds.

\subsection{Sleeping KW Reduction}  \label{subsec:sleepKW}

We give our own version to The {\em Kuhn-Wattenhofer Color Reduction (KW)} \cite{KW06} algorithm. This version is more efficient in terms of the sleeping model. The original algorithm starts with an initial coloring provided by the algorithm of Linial \cite{L86} and proceeds to reduce the number of colors in phases. There are $O(\log \Delta)$ phases, each takes $O(\Delta)$ time. In such a reduction, each vertex $v$ with color $l$ awaits the $l$ clock round and inspects the colors of its neighbors. Then $v$ chooses the smallest possible color available as its final color. In such a reduction the running time depends on the initial color palette size. In the KW algorithm these palettes are chosen such that they are of size at most $O(\Delta)$.

In terms of awake time the core issue is that a vertex $v$ requires information from all its neighbors potentially requiring its neighbors to be awake to transmit their colors to $v$. This means that there could be vertices that are kept awake only to report their colors to the vertices that color themselves at round $l$. So, we use an alternative. Instead of using the regular reduction we execute the algorithm $A$ to use the minimal colors possible to achieve the same result as the reduction. The algorithm $A$ depends on the size of the palette as well as for palette of size $d$ the performance of $A$ is $O(\log d)$ awake time and $O(d)$ clock rounds. Asymptotically, we do not require additional clock rounds than the original version of the KW algorithm but we do require less awake rounds in each phase in the KW reduction. Therefore, the number of clocks required for our version of KW reduction is still $O(\Delta \log \Delta)$. Since the KW algorithm has $O(\log \Delta)$ phases and each requires at most $O(\log \Delta)$ awake rounds, we achieve $O(\log^2 \Delta)$ awake time. We must not forget that the reduction starts with the coloring of Linial and therefore we have the following result.

\begin{thm}  \label{thm:KWsub}
Any {\bf O-LOCAL} problem can be solved in $O(\log^2 \Delta + \log^*n)$ awake time using at most $O(\Delta \log \Delta + \log^*n)$ clock rounds.
\end{thm}

\subsection{Preferring Low Awake Complexity} 

We give yet another version for solving problems in the {\bf O-LOCAL} class. We use the KW reduction again only this time we reduce the coloring of several sets of colors at once instead of one color at a time. In the previous version, at each phase of the KW reduction we achieved a palette of size $O(\Delta)$ for the reduction by reducing the colors of a union on pairs of color sets. Let $V_1,\dots,V_r$ be the color sets and let $\epsilon>0$ be some constant. Instead of reducing the colors on $V_1 \cup V_2, V_3 \cup V_4, \dots$, which means pairs of color sets, we reduce the colors on the union of $\Delta^\epsilon$ color sets. The reduction is done on the color sets of the union $V_1 \cup V_2 \cup \dots V_{\Delta^\epsilon}, V_{\Delta^{\epsilon+1}} \cup V_{\Delta^{\epsilon+2}} \dots$ and on each such union we execute the algorithm $A$. According to the performance property of the algorithm $A$, this requires $O((1+\epsilon) \log \Delta)$ awake rounds and $O(\Delta^{1+\epsilon})$ clock rounds. The number of phases we now require is $O(\frac{\log \Delta}{\log \Delta^\epsilon}) = O(\frac{1}{\epsilon})$. We still manage to achieve a $(\Delta+1)$-vertex-coloring. As mentioned before, this result improves on the clock round complexity achieved in \cite{BM21} without incurring additional awake complexity.

\begin{thm}  \label{thm:KWdebts}
Any {\bf O-LOCAL} problem can be solved in $O(\log \Delta + \log^*n)$ awake time using at most $O(\Delta^{1+\epsilon} + \log^*n)$ clock rounds where $\epsilon>0$ is a constant.
\end{thm}

\subsection{Preferring Clock Rounds}

In this section we offer another sleeping algorithm in which the number of clock rounds is lower if one is willing to pay in awake complexity. We note, though, that this payment is not significant compared to the complexity achieved in Section \ref{subsec:sleepKW}. In their monograph \cite{BEbook13} Barenboim and Elkin showed that the KW color reduction can be used in combination with defective coloring to create a series of $\log^*\Delta$ algorithms such that the overall running time of computing a $(\Delta+1)$-vertex-coloring is $O(\Delta + \log^*\Delta \log^* n)$ clock rounds. In that series of algorithms they used KW reduction and the KW iterative reduction in each algorithm for reducing the number of colors of the defective subsets of colors. In \cite{BEbook13} Barenboim and Elkin named {\em KW iterative reduction} the algorithm of which the reduction starts with a given coloring $\phi$. Given a coloring using $a$ colors, it is known that the KW reduction in this case requires $O(\log \frac{a}{\Delta+1})$ phases each in which a reduction is done on palette of size $O(\Delta)$. We can replace both these algorithms with our own versions of KW reduction for the sleeping model (as we only changed the technique of reduction in each phase, the number of phases remain the same in the KW iterative reduction. We call this the {\em Sleeping KW Iterative Reduction}). We do this by defining the algorithm $H_1$ to be the algorithm from Theorem \ref{thm:KWsub}. For integer values of $k \geq 2$ we define the algorithm $H_k$ in Algorithm \ref{alg:ak}. Ours is a revision of the algorithm in \cite{BEbook13}. This variation brings us closer to the lower bound of the {\bf O-LOCAL} class of problems while incurring a very small additive factor in the awake complexity in comparison to the algorithms in Section \ref{subsec:sleepKW}.

\begin{algorithm}[H] 
\caption{$\mathcal{H}_k(G)$}
\begin{algorithmic}[1]   \label{alg:ak}

\STATE $p = \log^{(k-1)} \Delta$.
\STATE Compute an $O(\frac{\Delta}{p})$-defective coloring using $p^2$ colors. Denote $G_1, \dots, G_{p^2}$ the subgraphs induced by the color sets. 

\FOR {each $G_i$ in parallel}
	\STATE Compute a proper coloring for $G_i$ using the algorithm $\mathcal{H}_{k-1}$. Denote the coloring as $\mu_i$.
\ENDFOR

\STATE Denote as $\mu$ the $O(p \Delta)$-coloring created from $\mu_1, \dots, \mu_{p^2}$. Invoke the sleeping KW iterative reduction on $\mu$ to gain a $(\Delta+1)$-coloring on $G$. Denote final coloring as $\phi$.

\STATE return $\phi$

\end{algorithmic}
\end{algorithm}

\begin{thm} \label{thm: BEvar}
Any {\bf O-LOCAL} problem can be solved in $O(\log^2 \Delta + \log^*\Delta \log^*n)$ awake time using at most $O(\Delta + \log^*\Delta \log^*n)$ clock rounds.
\end{thm}

\begin{proof}
In \cite{BEbook13} it is shown that the number of clock rounds required for the algorithm $H_{\log^*\Delta}$ is $O(\Delta + \log^*\Delta \log^*n)$. Our version of KW require $O(\Delta \log \Delta + \log^*\Delta)$ clock rounds which is the same as the standard KW reduction which is used in \cite{BEbook13}. Therefore the number of clock rounds is the same as in \cite{BEbook13}. Thus, we focus this proof on the awake complexity. The proof is by reduction on $2 \leq i \leq k$. Denote $T(H_i)$ the awake time of the algorithm $H_i$. We inspect the lines in Algorithm \ref{alg:ak}. Line 2 requires $O(\log^*n)$ awake rounds \cite{L86}. Line 4 requires $T(H_{i-1})$ awake rounds. Line 6 invokes the sleeping KW iterative reduction. Each phase in this variant takes $O(\log \Delta)$ time. There are $O(\log \frac{p\Delta}{\Delta}) = O(\log p) = O(\log^{(i)}\Delta)$ phases and overall line 6 requires $O(\log \Delta \log^{(i)}\Delta)$ awake rounds. Specifically, for every value of $2 \leq i \leq k$, line 6 requires at most $O(\log \Delta \log \log \Delta)$. We denote $c$ as a super-constant to represent the maximum between constants hidden in the awake times for all algorithms. Then,
$$T(H_{\log^*\Delta}) \leq c(\log \Delta \log \log \Delta + \log^*n) + T(H_{\log^*\Delta-1}) \leq $$
$$2c(\log \Delta \log \log \Delta + \log^*n) + T(H_{\log^*\Delta-2}) \leq $$
$$3c(\log \Delta \log \log \Delta + \log^*n) + T(H_{\log^*\Delta-3}) \leq \dots$$
$$i \cdot c(\log \Delta \log \log \Delta + \log^*n) + T(H_{\log^*\Delta-i}) \leq \dots$$
$$c(\log^*\Delta-1)(\log \Delta \log \log \Delta + \log^*n) + T(H_{1})$$ \\
From Theorem \ref{thm:KWsub} executing $H_1$ on a graph of degree $\frac{\Delta}{\log \Delta}$ gives us $T(H_1) = O(\log^2 \Delta + \log^*n)$. Thus, $T(H_{\log^*\Delta}) =$ $O(\log^2 \Delta + \log^*n) + O(\log \Delta \log \log \Delta \log^* \Delta + \log^*\Delta \log^*n) =$ $O(\log^2 \Delta + \log^*\Delta \log^*n)$. \\
Using this $(\Delta+1)$ coloring with the algorithm $A$ provides a solution to any {\bf O-LOCAL} problem in the same awake and running times as the algorithm $H_{\log^*\Delta}$.
\end{proof}

\section{Sleeping in Bounded Neighborhood Independence}

Given some acyclic orientation $\mu$ on an input graph $G$ with bounded neighborhood independence $K$, we give an algorithm that solves any {\bf O-LOCAL} problem in $O(K)$ awake deterministic time. The general scheme of the algorithm is quite simple. Using the orientation $\mu$, we build a new partial orientation $\phi$ on the edge set in the sense that some edges are not being orientated. We show that solving an {\bf O-LOCAL} problem $P$ using this new orientation is the same as solving $P$ on the original orientation. The difference is that the new orientation requires less awake rounds from vertices in $G$ which is our goal.

The new orientation is built as follows. At first all vertices of the graph are awake and each vertex is collecting the information from its 1-hop-neighborhood. Then all vertices enter a sleep state. Each vertex $v$ chooses parents from its 1-hop-neighborhood as follows. Internally, $v$ calculates an independent ruling set of the subgraph induced by its neighbors which have labels smaller than the label of $v$ and notes them as its parents. Moreover, these vertices are chosen such that their labels are the greatest among the potential parents of $v$. This can be done greedily by choosing the neighbor $u$ such that $ID(u) = \min(ID(w)\ |\ w \in N^*(v))$ and $ID(u) < ID(v)$ (where $N^*(v)$ is the 1-hop-neighborhood of $v$ without including $v$ itself). Then $v$ sets $u$ as a parent. $v$ then repeats this process until it can no longer choose new parents that holds the above requirements.
For solving a given problem $P \in $ {\bf O-LOCAL} each vertex is awakened in the clock round equal to each of its parents labels as well as the clock round that equals its own label. We now show that solving $P$ on $G$ using the orientation $\phi$ indeed gives us the solution of $P$ on $G$ according to $\mu$.

\begin{lem}
Let $v$ be a vertex in $G$ and $u$ be a neighbor of $v$ such that $\mu(u) < \mu(v)$. Then, in clock rounds $L(v)$, $v$ has the knowledge of the solution of $u$, denote $s(u)$, made. Thus, $v$ can resolve $P$ internally.
\end{lem}

\begin{proof}
We will show that there must be an oriented path in $\phi$ on which the value of $s(u)$ is propagated such that in round $L(v)$ the information of $s(v)$ is in $v$ and $v$ can use the information to calculate $s(v)$. 
In the case $u$ is a parent of $v$ in $\phi$ then $v$ is awake in clock round $L(u)$ and receives the value $s(u)$ for $P$ in that round and we are done. Otherwise, $u$ is not a parent of $v$ in $\phi$. Denote the subset of vertices $v$ chose as a ruling set from its neighbors as $M(v)$. When $v$ chose the set $M(v)$ then there must be a vertex $w_1 \in M(v)$ which $v$ chose over $u$. Thus, $w_1$ and $u$ are neighbors in $G$, and $L(w_1) > L(u)$. \\
If $u$ is a parent of $w_1$ in $\phi$ then $w_1$ is awake in clock round $L(u)$ and receives the knowledge of $s(u)$ in that round. Then $w_1$ passes this knowledge to $v$ in clock round $L(w_1)$ as $v$ is awake in that clock round. Otherwise, If $u$ is not a parent of $w_1$ in $\phi$ then there must be a vertex $w_2 \in M(w_1)$ which $w_1$ chose over $u$. Therefore, $L(w_1) > L(w_2) > L(u)$ and $w_2$ is also a neighbor of $u$ in $G$. The same as before, if $u$ is a parent of $w_2$ in $\phi$ then $w_2$ passes $s(u)$ to $w_1$ in clock round $L(w_2)$ and $w_1$ passes $s(u)$ to $v$ in clock round $L(w_1)$. \\
Since we assume the graph is finite and therefore any acyclic orientation on its edges is finite, an oriented path can only be of length of the orientation, denote $\ell$. Note that for each $i$ we have the following: First, $w_i$ is a neighbor of $u$ in $G$. Secondly, for each $1 \leq j \leq i$, any $w_j$ is a neighbor of $u$ and did not choose $u$ to be in $M(w_j)$. Thirdly, for each $i$ we have $L(w_i) < L(w_{i-1})$. We can have only so many repetitions of reducing the labels for the value of $i$ until we reach the label of $u$ itself. Thus, for some value of $i \leq \ell$, $w_i$ must choose $u$ to be in $M(w_i)$ and our oriented path is complete. it contains the vertices, in order, $u, w_i, w_{i-1}, \dots, v$. Each vertex in this path is awake in the clock round equals to the label of its parent allowing $s(u)$ to propagate to $v$ until clock round $L(v)$ arrives in which $v$ can use $s(v)$ to compute $s(v)$.
\end{proof}

\noindent From the above lemma we have that $v$ has all resolutions of its parents according to the orientation $\mu$ at clock round $L(v)$ in which $s(v)$ is calculated. Note that $v$ has at most $K$ parents in the orientation $\phi$ as its parents has to be independent of each other. According to the definition of the problem $P$ as a {\bf O-LOCAL} problem, we achieve the following result.

\begin{thm} \label{thm:givenOrient}
Given a graph with an acyclic orientation, any {\bf O-LOCAL} problem can be solved deterministically within $K+1$ awake rounds where $K$ is the maximal neighborhood independence of the input graph $G$.
\end{thm}

One now might wonder where this initial orientation comes from. One can use the algorithm and feed it with an orientation as an input as long as that orientation is acyclic. We use any of the algorithms given in Section \ref{sec:tradeOffs} to achieve this initial orientation with different awake complexities as well as different clock round complexities. In this section, though, we focus on the awake time complexity and prefer to neglect the number of clock rounds. And so we only need an acyclic orientation which one can achieve in a single awake round for all vertices by orienting the edges towards the vertex with the greater ID. This immediately gives us the following result and the corollary which follows. 

\begin{thm}  \label{thm:main}
Any {\bf O-LOCAL} problem can be solved deterministically within $K+1$ awake rounds where $K$ is the maximal neighborhood independence of the input graph $G$.
\end{thm}

\begin{col}
Any {\bf O-LOCAL} problem can be solved within constant deterministic awake time in graphs with bounded neighborhood independence.
\end{col}

\hfill
\hfill

\noindent {\bf A Remark on the Trade-Off Between Awake Time Complexity and Number of Clock Rounds Complexity:} \newline
Our algorithm is not dependent on the length of the given orientation in terms of the awake time complexity. But it is clear that the length of the orientation effects the communication clock rounds complexity. For an orientation of length $\ell$ we need $O(\ell)$ clock rounds. The orientation we use above can be of length $O(n)$. But if one is willing to trade off $O(\log^*n)$ awake time there are coloring algorithms that will guarantee a shorter orientation such as the algorithm of Linial \cite{L86} for $O(\Delta^2)$-vertex-coloring. This shows another strength of our algorithm as it can work with any algorithm for orienting the edges of the input graph. 

\begin{col}
Any {\bf O-LOCAL} problem can be solved within $O(\log^*n)$ awake time  and $O(\log^*n)$ clock rounds in graphs with bounded neighborhood independence.
\end{col}

\section{Dynamic Distributed {\bf O-LOCAL} Problems}  \label{sec:dynamicOLOCAL}

We devise a deterministic dynamic algorithm in the sleeping setting for maintaining a solution for a given {\bf O-LOCAL} problem $P$. This scheme is later improved significantly but lays the basis for understanding how to use coloring algorithms as dynamic sleeping algorithms for {\bf O-LOCAL} problems. The preparation stage is the algorithm in \cite{BM21} for solving {\bf O-LOCAL} problems, denoted $A$. We do this only at the preparation stage. Given a problem $P$ as defined above one can solve $P$ in $O(\log \Delta)$ awake time with $O(\Delta)$ number of clock rounds. \\
Let $t$ be an integer representing the number of changes between executions of the update algorithm. That is, we advance in phases. In each phase the vertices of the graph are asleep for $t$ clock rounds and wake up after to check if changes were made to their adjacent edges. Without loss of generality, we allow for a single change in each of the $t$ clock rounds in the change state. This includes both changes to edges and vertices. We start with the assumption the problem $P$ was solved on $G$ prior to the current changes. Then we run an update algorithm during which there can be no changes to the graph. Our algorithm recolors the sub-graph which was changed properly and uses $A$ to maintain the solution for the problem $P$ on the new version of $G$ thus keeping the assumption we started from for the next batch of changes. \\
Let $S$ be the subset of vertices whose adjacent edges changed and also including new vertices. Let $k = \min(t, \Delta)$. Notice that $|S| = O(t)$. We set the color of the vertices in $S$ to 0 (no color). We color $S$ using $O(k^2)$ colors in $O(\log^*t)$ time. We then awake the entire graph for 1 round so that vertices in $S$ can collect the colors of vertices in $N(S)$, that is the information from the vertices in $G$ of distance at most 1 from vertices in $S$, as well as their yet legal choices in regard to problem $P$. We then use the algorithm $A$ to reduce the number of colors in $S$ to $k+1$. We now again use the algorithm $A$ to Solve $P$ within $S$ when the orientation we use is one where vertices outside of $S$ are directed into $S$ and inside $S$ the orientation continues according to the coloring.

\begin{lem} \label{lem:partialS}
The new orientation allows us to solve $P$ legally on $G$.
\end{lem}

\begin{proof}
In the subgraph induced on the vertex subset $V/S$ the problem $P$ is already solved and legal. One can now regard the process of solving $P$ on $S$ as if we activated $A$ on $G$ where vertices in $V/S$ take priority in the orientation. But these vertices do not awake in our scheme besides for a single round to have their neighbors in $S$ know of their choices. From now solving $P$, since it is an ${O-LOCAL}$ problem, is a manner of progressing on the orientation inside $S$ which we are given by the coloring of $S$.
\end{proof}

\noindent Note that although overall in $G$ there might now be more colors than $(\Delta+1)$ we do not use all these colors in our scheme as only vertices in $S$ need recalculation for $P$. We do though preserve the property that $P$ is solved on $G$ and that is all we require for Lemma \ref{lem:partialS}.

\begin{thm}
Let $t$ be the number of edge additions that occur between executions of an update algorithm. Let $\alpha = \min(\Delta, t)$ and let $\beta = \min(n, t)$ where $\Delta$ is the current maximum degree of the graph. Then there is a deterministic dynamic sleeping algorithm for solving any problem in the class of {\bf O-LOCAL} problems with awake update-time $O(\log \alpha + \log^*\beta)$. The update requires $O(\alpha^2 + \log^* \beta)$ clock rounds.
\end{thm}

Using the results from Table \ref{tab:OLOCAL}, each as a dynamic algorithm in the same manner presented in this section, we can achieve three results each expressed in terms of the changes made to a graph and using the algorithms in question as the coloring algorithm for the subset $S$. The running times can be concluded easily in terms of $\alpha, \beta$. One simply sets $\Delta = \alpha$ and $n = \beta$ in Table \ref{tab:OLOCAL} to achieve the complexities of the dynamic variations. \\

\section{Dynamic Algorithm for Decidable Problems}

In \cite{BM21} Barenboim and Maimon offered an algorithm for building a Distributed Layered Tree (DLT) in $O(\log n \log^*n)$ awake time that requires $O(n \log n \log^*n)$ clock rounds. Loosely described, a DLT is an oriented tree in $G$ with a labeling to the vertices of the graph. The labeling is such that the label of each vertex is greater than that of its parent in the tree. A close inspection of this algorithm shows that these complexities derive from the initial $n$ vertices that are connected to a single DLT. We show here that the algorithm can be adapted to the dynamic setting. Assume that in the preparation stage we built a DLT for $G$. The following lemma shows the number of sub-DLTs created by the changes to the graph.

\begin{lem}
After $t$ changes to $G$, a DLT of $G$ is partitioned to at most $O(t \Delta)$ distinct sub DLTs.
\end{lem}
\begin{proof}
We prove this by types of changes. Adding an edge or a vertex does not shatter the existing DLT of a connected component. Adding an edge, though, does have the potential to unite two connected components into one, thus requires us to unite 2 DLTs into one. Thus, addition creates at most 1 connection requirement between the existing DLTs. Removing an edge, if it one that is in the DLT, shatters the DLT into two, again requires us for at most 1 connection between two DLTs, amending the disconnection (if such amend is possible). Removing a vertex removes all $\Delta$ adjacent edges thus requiring at most $\Delta$ connections between $\Delta+1$ sub DLTs to re-unite into a single DLT (where possible). Therefore, the worse case is that we will have $t$ vertex removal and from this we conclude that we have at most $O(t \Delta)$ sub DLTs to connect into a single DLT that spans the input graph.
\end{proof}

\noindent Thus, we are required for $O(\log t + \log \Delta)$ connection phases using the same DLT algorithm from \cite{BM21}. Each phase requires coloring in $O(\log^* t + \log^*\Delta)$ awake time. The number of clock rounds is still $O(n)$ in each connection phase. 

\begin{thm}
Let $t$ be the number of changes between executions of an update algorithm. Let $b = \min(\max(t, \Delta), n)$. In the dynamic sleeping model there is a deterministic algorithm for any decidable problem with awake update-time of $O(\log b \log^*b)$. The update requires $O(n \log b \log^*b)$ clock rounds.
\end{thm}

\noindent {\bf Remark:} In the worst case where the number of changes is very large, i.e. $\Omega(n)$, our algorithm aligns with the algorithm at \cite{BM21}.

\bibliography{lipics-v2019-sample-article}

\end{document}